\documentclass{article}
\usepackage[linesnumbered]{algorithm2e}
\usepackage{amsmath}
\usepackage{amsfonts}
\usepackage{amssymb,amsthm}
\usepackage{color}
\usepackage[margin=1in]{geometry}
\usepackage{hyperref}
\usepackage{graphicx}
\usepackage{verbatim}

\numberwithin{equation}{section}
\newtheorem{theorem}[equation]{Theorem}
\newtheorem{proposition}[equation]{Proposition}
\newtheorem{lemma}[equation]{Lemma}

\newtheorem{definition}[equation]{Definition}

 \newcommand{\m}[1]{\left|#1\right|}

 \author{Sheela Devadas\thanks{Math Dept., MIT, Cambridge MA 02139
 E-mail: {\tt sheelad@mit.edu}} \;and Ronitt Rubinfeld\thanks{CSAIL, MIT, Cambridge MA 02139 and
the Blavatnik School of Computer Science, Tel Aviv University.
E-mail: {\tt  ronitt@csail.mit.edu}.
Research supported by NSF grants CCF-1217423, CCF-1065125, CCF-1420692, and ISF grant 1536/14.}}
\title{A Self-Tester for Linear Functions over the Integers with an Elementary Proof of Correctness}
\date{June 22, 2015}

\begin{document}

\maketitle

\begin{abstract}
We present simple, self-contained proofs of correctness for algorithms for 
linearity testing and program checking of linear functions on finite subsets of integers represented as $n$-bit numbers. 
In addition we explore a generalization of self-testing to homomorphisms on a multidimensional vector space.
 We show that our self-testing algorithm for the univariate case can be directly generalized to vector space domains.
The number of queries made by our algorithms is independent of domain size. 
\end{abstract}

\section{Introduction}

In this paper, we consider the problem of linearity testing,
as in the program checking, self-testing and property testing
frameworks of \cite{BK},\cite{BLR},\cite{RS96},\cite{GGR}. A function $f$ is 
{\em linear} if for any $x,y$ in the domain, $f(x)+f(y)=f(x+y)$. In the case where the domain and range are both the set of integers, 
it is easily shown (by induction) that 
the linear functions are exactly the functions
that  multiply by a constant: i.e., $f(x)=bx$ for some $b$. 
Given a program $P_b$ that computes multiplication by $b$,  
our self-tester should pass $P_b$ if it gives the correct answer 
for all inputs, and should fail $P_b$ 
with high probability if it is incorrect on a large enough
fraction of the $n$-bit integer inputs. 

We present efficient algorithms for linearity testing
and program checking of functions on integers
represented as $n$-bit numbers
as considered in \cite{BLR}. Our query complexities are of the same
order of magnitude, though the constants we achieve are not
as good as in \cite{BCHKS},\cite{KLX},\cite{Kthesis}. 
   However, our proofs are 
more elementary;   in contrast
to previous proofs, our proofs use little ``algebraic structure''.
We then show that our techniques can be extended to homomorphisms on a multidimensional vector space.

\paragraph{Previous Work }

Linearity testing was first considered in \cite{BLR}
and used to give algorithms for self-testing
and program checking programs that compute linear functions. 
Linearity testing for multivariate functions over various finite groups 
has been considered in several works, including \cite{BCHKS},\cite{KLX},\cite{BFL},\cite{BGLR},\cite{Trevisan},\cite{ST},\cite{ST2},\cite{HW},\cite{SW},\cite{BSVW},\cite{SudanTrevisan}. In \cite{BLR},\cite{RS92}, linearity testing for multivariate functions over finite subsets of infinite groups is considered. Testers for real-valued multilinear functions defined over finite domains have been studied in \cite{Magniez} 
. Linearity testing for functions over finite subsets of rational domains is considered in \cite{HW},\cite{RS92}. 

Several works have considered the problem of testing low-degree polynomials. Specifically relevant to this work are results on testing polynomials over rational domains, such as \cite{RS96}. 


\paragraph{Outline of Paper}

In Section 2, we give definitions of 
property testing, self-testing, and checking algorithms from \cite{BK},\cite{BLR},\cite{RS96},\cite{GGR},
and give
a general overview of the other 
definitions and techniques necessary for our proofs. 
In Section 3, we present self-testing and property testing algorithms for the univariate case and proofs 
that they catch errors in all
programs whose output function 
differs 
from any linear function 
in a significant fraction of locations.
In Section 4, 
we discuss how to extend the results we have for the univariate case to linear homomorphisms on a vector space.
Finally, in the Appendix we present an additional checking 
algorithm for the univariate case. 

\section{Preliminaries}

A {\em checking algorithm}, as defined in \cite{BK} for a function $f$, is an algorithm $C$ that gets
as input a specific input $x$ and program $P$. The goal
is to determine whether program $P$ is correct on input $x$ or
has a fault.   $C$ may make calls to $P$ on any input.
If $P=f$ for all inputs, the checking algorithm $C$ should PASS with
probability at least $2/3$; 
but if $P(x) \ne f(x)$ at the chosen input $x$, 
algorithm $C$ returns FAIL with probability $2/3$, where the probability is over the coin tosses of $C$, and not on any assumption on the input distribution.\footnote{Note
that by repeating $C$ at least $1/\epsilon$ times
and outputting the majority answer, error
probability $\leq \epsilon$ can be achieved.}
Note that if $P(x)=f(x)$, but there is a $y \neq x$ such that
$P(y) \neq f(y)$, then $P$ may either output PASS or FAIL.

A {\em self-testing algorithm}, as defined in \cite{BLR} for a function $f$ 
over a finite domain $D$, is an algorithm $T$ that gets
as input a program $P$ and a parameter $\epsilon$. The goal
is to determine whether program $P$ is correct on most inputs in $D$.   
$T$ may make calls to $P$ on any input.
If $P=f$ for all inputs in $D$, the self-testing algorithm $T$ returns PASS with
probability at least $2/3$; 
but if $P(x) \ne f(x)$ on $\ge \epsilon$ fraction of inputs in $D$, 
algorithm $T$ returns FAIL with probability $2/3$.
Note that if there exists an $x$ such that $P(x)\ne f(x)$, 
but $P(x) \ne f(x)$ on at most $\epsilon$ fraction of inputs in $D$, 
then $T$ may either output PASS or FAIL. 

In general, given a finite domain $D$, we say that $P$ is {\em $\epsilon$-close} to $f$ if the probability that $P(x) = f(x)$ over $x$ chosen uniformly from $D$ is $\ge 1-\epsilon$, and that $P$ is {\em $\epsilon$-far} from $f$ if it is not $\epsilon$-close. Therefore our self-testing algorithm returns FAIL with high probability ($2/3$) if $P$ is $\epsilon$-far from $f$. 

A {\em property testing algorithm} for a function family $F$,
as defined in \cite{RS96} 
(there described as an $\epsilon$-function-family-tester),
is an algorithm T that gets as input a program $P$ and a parameter $\epsilon$. 
$T$ may make calls to $P$ on any input. 
If for some $f \in F$, $P=f$ on all 
inputs, the testing algorithm $T$ returns PASS with 
probability at least $2/3$, but if $P$ is $\epsilon$-far from all $f \in F$, 
then $T$ returns FAIL with probability $2/3$. 

For checking, self-testing, and property testing algorithms, the parameters to optimize are the number of queries to the program $P$ and the additional computation time the algorithm needs to perform, where we define the {\em additional computation time} as the running time of the algorithm not including time spent by $P$ in answering queries. In this paper we focus on algorithms that make a constant number of queries to $P$ and incur an additional computation time that is only linear in the input size. Note that for functions such as integer multiplication, no linear time algorithm is known; the best known algorithm is $O(n \log n \log \log n)$, in \cite{Furer}. 
Thus it is not known how to test multiplication by comparing results
with another known multiplication program using only an additional
cost of linear time.

Here we focus on programs $P$ 
purporting to compute not just some linear function but a specific linear function $f$: for these programs, we give a self-testing algorithm. In the univariate case, 
such a function must be $f(x)=bx$ for some $b$. 
A similar constraint is true for the case of a linear homomorphism on a multidimensional vector space.
We rely on being able to compute $f(x)$ more quickly for inputs $x$ of a specific form: in particular, $b \cdot 2^n$ can be computed in linear time by shifting $b$ by $n$ spaces. 
Note that multiplication and division by $2^n$ can be done in time linear in $n$.

\section{Testing Algorithm}\label{section:tester}

We describe an algorithm that tests the correctness of a program $P_b$ that purports to compute a function $f_b(x) \dot{=} b \cdot x$. Let $P_b(x)$ indicate the program's output when given the input $x$. We refer to the domain of integers represented as $n$-bit numbers as $D_n$; the size of $D_n$ is $2^n$. 

\begin{theorem}\label{thm:testingalg}
There exists a self-testing algorithm for the linear function $f_b$ over the domain $D_n$ with $O(n/\epsilon)$ additional computation time and $O(1/\epsilon)$ queries to the program $P_b$.
\end{theorem}

We will also describe an algorithm that tests the correctness of $P_b$ on a specific input:

\begin{theorem}\label{thm:checker}
There exists a checking algorithm for the linear function $f_b$ over the domain $D_n$ with $O(n/\epsilon)$ additional computation time and $O(1/\epsilon)$ queries to the program $P_b$.
\end{theorem}

The proof of Theorem~\ref{thm:checker} is given in the appendix.

To prove Theorem~\ref{thm:testingalg}, we first define a function RandSplit that is useful in our testing and checking algorithms and will be redefined in the case of a homomorphism on a vector space. The idea of this function is to use the distributive property of multiplication: $b(a+c)=ba+bc$. It should be the case that $P_b(a+c)=P_b(a)+P_b(c )$, so we will check that this is the case for random $a$ and $c$. One issue is that
$a+c$ can be bigger than $2^n$, though not bigger than $2^{n+1}$.   We
can keep track of which is the case, and
if $(a+c)= 2^n+ x$  we can check that $P_b(a)+P_b(x) + b \cdot 2^n = P_b(a+c)$ instead. We can use this `wraparound' property and our ability to calculate $b \cdot 2^n$ easily, 
via a linear time shift operation, 
to verify that $P_b$ satisfies this distributive property.

\IncMargin{1em}
\begin{algorithm}[H]

\SetKwFunction{Rand}{Rand}\SetKwFunction{RandSplit}{RandSplit}
\SetKwInOut{Input}{input}\SetKwInOut{Output}{output}
\Input{An $n$-bit number $x$, an integer $b$, a program $P_b$ for computing multiplication by $b$}
\Output{FAIL or PASS}
$x_1 \leftarrow$ \Rand(0,$2^n$)\;
\lIf{$x_1 < x$}{ $\delta \leftarrow 0$}
\lElse{ $\delta \leftarrow 1$}
$x_2 \leftarrow \delta \cdot 2^n + x - x_1$\;
\If{$P_b(x_1) + P_b(x_2) \ne b \cdot  \delta \cdot 2^n + P_b(x)$}{
\Return FAIL\;
}
\Return PASS\;
\caption{{\sc RandSplit}}
\label{algo:RandSplit}
\end{algorithm}\DecMargin{1em}

We now use Algorithm~\ref{algo:RandSplit} to create a testing algorithm. Before running Algorithm~\ref{algo:RandSplit} at least $k_2$ times, we test $n$-bit inputs summing to $2^n$ for the same distributive property $k_1$ times. The combination of these two tests will allow us to detect errors in $P_b$. We will show that it suffices for $k_1+k_2$ to be set to $O(1/\epsilon)$. 

\IncMargin{1em}
\begin{algorithm}[H]
\SetKwFunction{Rand}{Rand}\SetKwFunction{RandSplit}{RandSplit}
\SetKwInOut{Input}{input}\SetKwInOut{Output}{output}
\Input{A program $P_b$ for computing multiplication by $b$}
\Output{FAIL or PASS}

\For{$i\leftarrow 1$ \KwTo $k_1$}{ \label{algo:Test:fflst}
$x \leftarrow$ \Rand(0,$2^n$)\;
\If{$P_b(x) + P_b(2^n-x) \ne b \cdot 2^n$}{  \label{algo:Test:fflif}
\Return FAIL\; \label{algo:Test:ffl}
}
} \label{algo:Test:fflend}

\For{$i\leftarrow 1$ \KwTo $k_2$}{ \label{algo:Test:sflst}
$x \leftarrow$ \Rand(0,$2^n$)\; \label{algo:Test:sfl} 
\If{\RandSplit$(x,b,P_b)$=FAIL}{
\Return FAIL\; \label{algo:Test:sfl2}
}
} \label{algo:Test:sflend} 
\Return PASS\;
\caption{{\sc Test}}
\label{algo:Test}
\end{algorithm}\DecMargin{1em}

\begin{lemma}[Main Lemma]\label{lemma:main} If $P_b$ is correct on $\le1-\epsilon$ fraction of inputs in $D_n$, then Algorithm~\ref{algo:Test} returns FAIL with probability $\ge 3/4$. \end{lemma}

\begin{proof}

We first define the discrepancy of an input:

\begin{definition} The discrepancy of $x$ is $d(x) \dot{=} P_b(x) - x \cdot
b$. \end{definition}

If $P_b(x)$ is correct, then $d(x)=0$. By our assumption, $d(x) \ne 0$ for at least $\epsilon$ fraction of the inputs; we let $\epsilon_0$ be the actual fraction of inputs for which $d(x) \ne 0$. We see that $\epsilon$ is the input to the program and $\epsilon_0$ is the actual error rate of the program. What we wish to show is that when $\epsilon_0 \ge \epsilon$, then we return FAIL with high probability.

We now show that if a function is likely to pass the test, the discrepancy function must have a certain form; the number of inputs with positive discrepancy and the number of inputs with negative discrepancy must be about the same. 

\begin{proposition} If less that $\epsilon_0/2-\beta$ fraction of the numbers in the domain $D_n$ have discrepancy $>0\; (<0)$ then Line~\ref{algo:Test:ffl} will output FAIL with probability at least $2\beta$. \end{proposition}

\begin{proof}

Consider ordered pairs of the form $(x,2^n-x)$. The test in Line~\ref{algo:Test:fflif} pairs the numbers $(x,2^n-x)$ and detects an error if their discrepancies do not sum to $0$. Therefore if no error is detected but $x$ has positive discrepancy, $2^n-x$ must have discrepancy strictly less than $0$.
If $\le \epsilon_0/2-\beta$ fraction of the numbers have discrepancy strictly less than $0$, then since we know that $1-\epsilon_0$ fraction of the numbers have discrepancy $0$ (since $P_b$ is correct on those numbers), we see that $\ge \epsilon_0/2+\beta$ of the numbers have positive discrepancy. When pairing the numbers into $(x,2^n-x)$, we see that an error will be detected if the discrepancies do not sum to $0$.  However, by our assumption, there are fewer numbers with discrepancy strictly less than $0$; specifically, at least $\epsilon_0/2+\beta-(\epsilon_0/2-\beta)=2\beta$ fraction of the numbers more have $d(x)>0$ than have $d(x)<0$. Therefore $2\beta$ of the $(x, 2^n-x)$ pairs have $d(x) > 0$ and $d(2^n-x)\ge 0$. Therefore one iteration of the first `for loop' gives an error with probability $\ge 2\beta$ because the discrepancies need to sum to 0. The same proof works for showing that $> \epsilon_0/2-\beta$ fraction of the numbers must have discrepancy strictly greater than $0$ or there will be an error with probability $2\beta$. 
\end{proof}

This means that we can assume that if the loop in Lines~\ref{algo:Test:fflst}-\ref{algo:Test:fflend} in Algorithm~\ref{algo:Test} passes with probability $2\beta$, then the fraction of numbers with discrepancy greater than $0$ is $> \epsilon_0/2-\beta$ but $< \epsilon_0/2+\beta$. The same is true for the fraction of numbers with discrepancy strictly less than $0$. 

We now define a $\delta_x$ function related to the discrepancy function.

\begin{definition} $\delta_x(x_1) = 1$ if $d(x_1) > 0 $ and $d((x-x_1) \bmod 2^n) <
0$.  $\delta_x(x_1) = 0$ otherwise. \end{definition}

\begin{definition} The {\em number of opposite-sign matches} is 
$\sum_{0 \leq x_1 \leq 2^n} \delta_x(x_1)$. \end{definition}

\begin{proposition}
The expected {\em number of opposite-sign matches}
is $\leq (\epsilon_0/2)^22^n$. \end{proposition}

\begin{proof}
We assume that the fraction of inputs with $d(x) > 0$ is $\epsilon_1$, and the fraction of inputs with $d(x) < 0$ is $\epsilon_2$. We note that $\epsilon_1+\epsilon_2=\epsilon_0$. 
Using the fact that $x_1$ gets matched to any $x_2$ with 
probability
$1/2^n$, we see that the probability that $x_1$ has positive discrepancy and $x_2$ has negative disrepancy is just $\epsilon_1\epsilon_2$, so the expected number of opposite-sign matches is $\epsilon_1\epsilon_22^n\le(\epsilon_0/2)^22^n$. 
\end{proof}

We let $0 < \alpha < 1$ be a parameter that we will set later.

\begin{definition} $x$ is {\em good} if the number of opposite-sign matches is 
$\leq \frac{1}{\alpha} (\epsilon_0/2)^22^{n}$ and {\em bad} otherwise. \end{definition}

\begin{proposition} The probability that a bad $x$ is picked in Line~\ref{algo:Test:sfl} is
$\le \alpha$. \end{proposition}

\begin{proof} Apply Markov's inequality to the definition of a good $x$. \end{proof}

\begin{proposition}\label{prop:gooderr} If a good $x$ is picked in Line~\ref{algo:Test:sfl}, then with probability
$\ge \epsilon_0/2 - \beta- \frac{1}{\alpha} (\epsilon_0/2)^2$,  Line~\ref{algo:Test:sfl2} outputs FAIL. \end{proposition}

\begin{proof}
The probability of picking any $(x_1,x-x_1)$ pair during verification is $\frac{1}{2^n}$.

If we have $d(x_1)>0$ and $d(x_2)\ge0$,  then $d(x_1)+d(x-x_1)>0$. We see that the probability that $d(x_1)>0$ is $>\epsilon_0/2-\beta$ and the number of matches with $d(x_1)>0$ and $d(x-x_1)<0$ is $\le \frac{1}{\alpha} (\epsilon_0/2)^2$. Therefore the probability that $d(x_1)>0$ and $d(x-x_1)\ge0$ is $\ge \epsilon_0/2-\beta - \frac{1}{\alpha} (\epsilon_0/2)^2$. A similar proof shows that the probability that $d(x_1)<0$ and $d(x-x_1)\le0$ is also $\ge \epsilon_0/2-\beta -  \frac{1}{\alpha} (\epsilon_0/2)^2$. 

In the first case we note that $d(x_1)+d(x-x_1)>0$, while in the second case $d(x_1)+d(x-x_1)<0$. Therefore the probability that $d(x_1)+d(x-x_1)>0$ is $\ge \epsilon_0/2-\beta -  \frac{1}{\alpha} (\epsilon_0/2)^2$, as is the probability that $d(x_1)+d(x-x_1)<0$. Since the test only passes if $d(x)=d(x_1)+d(x-x_1)$ and $d(x)$ cannot be both positive and negative, a mistake must be found with probability at least $\ge \epsilon_0/2-\beta -  \frac{1}{\alpha} (\epsilon_0/2)^2$.
\end{proof}

Since we are given the value of $\epsilon$ and that $\epsilon_0 \ge \epsilon$, and we know the probability of finding an error in each part of the algorithm, we need only repeat the parts enough times to catch the error in order to output FAIL with high probability. 
\end{proof}

\subsection{Putting It Together}\label{sec:puttingittogether}

We now consider possible specific values for $\alpha,\beta$. Let $\beta=\epsilon/4$ and $\alpha=2/3$. 

What is the runtime of Algorithm~\ref{algo:Test}? Applying the Chernoff bound as given in the Appendix, we see that if we want to be able to say that $\ge \epsilon_0/2-\beta$ of the inputs have discrepancy $<0$ or $>0$ with probability $\ge 7/8$ we need to have $k_1=O(1/\beta)=O(1/\epsilon)$. The probability of detecting an error in $P_b$ in the loop in Lines~\ref{algo:Test:sflst}-\ref{algo:Test:sflend} is $\ge (1-\alpha)(\epsilon_0/2-\beta -  \frac{1}{\alpha} (\epsilon_0/2)^2) = \frac{\epsilon_0/4-3\epsilon_0^2/8}{3}$; we therefore need to run the loop $k_2=O(\frac{1}{2\epsilon-3\epsilon^2})$ times in order to expect to see this error with probability $7/8$; for sufficiently small $\epsilon$, this is $O(1/\epsilon)$ as well. Therefore $k_1+k_2$ is $O(1/\epsilon)$ as desired. If we see errors in both parts of the algorithm with probability $7/8$, by a union bound the probability that we fail to catch an error in $P_b$ in both parts is $1/4$, so we output FAIL with probability $3/4$ as desired.

If $\epsilon = 1/8$, then using the specific values for the Chernoff bound given in Theorem~\ref{thm:chernoff} in the Appendix, we see that $k_1=96$ and $k_2\approx 709$ 
will give us our desired result: Algorithm~\ref{algo:Test} will return FAIL with probability $\ge 3/4$.

\begin{proof}[Proof of Theorem~\ref{thm:testingalg}]
If $P_b$ is correct on all inputs in $D_n$, then Algorithm~\ref{algo:Test} will pass. 
If $P_b(x)\neq f(x)$ on at least $\epsilon$ fraction of the inputs in $D_n$,
then we have shown in Lemma~\ref{lemma:main} 
that Algorithm~\ref{algo:Test} returns FAIL with probability $\ge 3/4$. 
We have also shown that setting $k_1+k_2$ to be $O(1/\epsilon)$ lets Algorithm~\ref{algo:Test} will return FAIL with probability $\ge 3/4$, so the number of queries to  $P_b$ will also be $O(1/\epsilon)$. 
The extra computation done by the algorithm consists only of linear time operations such as the shift allowing us to compute $b \cdot 2^n$; 
therefore the additional computation time is $O(n(k_1+k_2))=O(n/\epsilon)$. 
Thus Algorithm~\ref{algo:Test} is a valid self-testing algorithm for the function $f_b$. 
\end{proof}

\subsection{General Linear Function}

We include a property testing algorithm for the case where we do not know what linear function $P$ claims to compute; only that it claims to compute a linear function. In this section, the property tester needs to pass any linear function and fail any function that is not $\epsilon$-close to some linear function. 

\begin{theorem}\label{thm:generallineartest}
There exists a property testing algorithm for linear functions on $n$-bit inputs with $O(n/\epsilon)$ additional computation time and $O(1/\epsilon)$ queries to the program $P_b$.
\end{theorem}

The idea of this general linear testing algorithm, or property testing algorithm, is similar to Algorithm~\ref{algo:Test} above in that we run two tests several times each. The first test here is to check that for various pairs of inputs $x,2^n-x$ that sum to $2^n$ that $P(x)+P(2^n-x)$ is always equal to $P(2^n)$, which must be $b \cdot 2^n$ for some integer $b$. This will allow us to reduce to the case of a specific linear function for the second part of this test, which is just running Algorithm~\ref{algo:RandSplit} on various inputs $x$ for the value of $b$ that we found in the first part. In much the same way as above, this will allow us to detect errors in $P$. We note that this algorithm essentially learns the value of $b$. 

\IncMargin{1em}
\begin{algorithm}[H]
\SetKwFunction{Rand}{Rand}\SetKwFunction{RandSplit}{RandSplit}
\SetKwInOut{Input}{input}\SetKwInOut{Output}{output}
\Input{A program $P$ claiming to compute a linear function}
\Output{FAIL or PASS}

$a \leftarrow P(2^n)$\;

\If{$2^n \nmid a$}{ \label{algo:GeneralLinearTest:ft}
\Return FAIL\;
}

\For{$i\leftarrow 1$ \KwTo $k_1$}{
$x \leftarrow$ \Rand(0,$2^n$)\;
$a_i \leftarrow P(x)+P(2^n-x)$\;
\If{$a_i \ne a$}{
\Return FAIL\;
}
}

$b \leftarrow a/2^n$\;

\For{$i\leftarrow 1$ \KwTo $k_2$}{
$x \leftarrow$ \Rand(0,$2^n$)\;
\If{\RandSplit$(x,b,P)$=FAIL}{
\Return FAIL\;
}
}
\Return PASS\;
\caption{{\sc GeneralLinearTest}}
\label{algo:GeneralLinearTest}
\end{algorithm}\DecMargin{1em}

\begin{lemma}\label{lemma:generallineartesthelper}
If $P$ is $\epsilon$-far from linear, then Algorithm~\ref{algo:GeneralLinearTest} returns FAIL with probability $\ge 3/4$. 
\end{lemma}

\begin{proof}
If $P$ passes the test in Line~\ref{algo:GeneralLinearTest:ft} that $2^n \mid P(2^n)$ we can reduce to the case where $P=P_b$ for $b=a/2^n$, and the testing algorithm is equivalent to the original testing algorithm. Since $P$ is $\epsilon$-far from linear, it is $\epsilon$-far from the function $b \cdot x$, so then by Lemma~\ref{lemma:main}, since Algorithm~\ref{algo:Test} will return FAIL with probability $\ge 3/4$, so will Algorithm~\ref{algo:GeneralLinearTest}.
\end{proof}

\begin{proof}[Proof of Theorem~\ref{thm:generallineartest}]
If $P$ is linear, Algorithm~\ref{algo:GeneralLinearTest} will output PASS with probability $1$. 
From Lemma~\ref{lemma:generallineartesthelper} 
it is clear that Algorithm~\ref{algo:GeneralLinearTest} 
will output FAIL with probability $\ge 3/4$ if we set $k_1+k_2=O(1/\epsilon)$ as in Algorithm~\ref{algo:Test}.
We see that it makes the same number of queries and has extra running 
time on the same order as Algorithm~\ref{algo:Test}. 
Therefore it is a linearity property tester
with additional computation time $O(n/\epsilon)$ and queries $O(1/\epsilon)$.
\end{proof}

\section{Multivariate Linear Functions}

In this section we give self-testers for linear homomorphisms on a vector space. For a vector space $V$ of dimension $m$, we say $f$ is a {\em linear homomorphism} on $V$ if for any ${\bf x_1},{\bf x_2} \in V$, $f({\bf x_1}) + f({\bf x_2}) = f({\bf x_1} + {\bf x_2})$. An example of a linear homomorphism in two variables is $f(x,y)=x+y$. 

In an $m$-dimensional vector space $V$, we let ${\bf e_i}$ be the vector that has $0$ for all coordinates but the $i$th, which is $1$. Using linear algebra, we see that any linear homomorphism on $V$ is determined by its values on the ${\bf e_i}$ - in fact, if we let $b_i=f({\bf e_i})$ for all $i$, then for any vector $\langle x_1,\dots ,x_m \rangle$ we see that $f(\langle x_1,\dots,x_m \rangle) = \sum_{i=1}^m b_ix_i$. 

To test multivariate linear homomorphisms, we make the assumption that we know $b_i=f({\bf e_i})$ for all $i$. Note that this is a generalization of the univariate case, in which we need to assume we know the value of $f(2^n)=b\cdot 2^n$, which means we know $b$. 

Let $P$ be a program that purports to compute $f$. We can modify the algorithm above to replace $b$ with the $b_i$ and modify the random-split function to split one vector into a random pair of vectors (by using the usual random-split function component-wise). Instead of calling Algorithm~\ref{algo:RandSplit} on a number $x$, integer $b$, and program $P$, we call it on a vector ${\bf x}$ and program $P$ and verify that for a random vector ${\bf y}$ we have $f({\bf y}) + f({\bf x} - {\bf y})=f({\bf x})$. 
Then the proof above still holds if we replace $2^n$ with the number of vectors in the vector space - 
which is $2^{mn}$ if our domain is
$m$ dimensional vectors of $n$-bit numbers. 
Because there was no dependence on the size of the domain, 
the same error bounds hold; therefore we can get the same bound that if $P$ is correct on $\le 3/4$ of the inputs, the program will return FAIL with high probability. 

We therefore have the following theorem given a linear homomorphism $f$ on a $m$-dimensional vector space $V$ with the values of $b_i=f({\bf e_i})$ known. We assume vectors in $V$ have coordinates that are $n$-bit integers. Therefore $\m{V}=2^{mn}$. 

\begin{theorem}\label{thm:linhomtestingalg}
There exists a self-testing algorithm for the linear homomorphism $f$ over the domain $V$ with $O(nm/\epsilon)$ additional computation time and $O(1/\epsilon)$ queries to the program $P$.
\end{theorem}

We redefine the algorithms for this case. 

\IncMargin{1em}
\begin{algorithm}[H]

\SetKwFunction{Rand}{Rand}\SetKwFunction{RandSplit2}{RandSplit2}
\SetKwInOut{Input}{input}\SetKwInOut{Output}{output}
\Input{An $m$-dimensional vector of $n$-bit numbers ${\bf x}=\langle x_1,\dots,x_m \rangle \in V$, the values $b_1,\dots, b_m$, a program $P$ for computing the function $f$}
\Output{FAIL or PASS}
\For{$i \leftarrow 1$ \KwTo $m$}{
$y_i \leftarrow $\Rand(0,$2^n$)\;
\lIf{$y_i < x_i$}{ $\delta_i \leftarrow 0$}
\lElse{ $\delta_i \leftarrow 1$}
$z_i \leftarrow \delta_i \cdot 2^n + x_i - y_i$\;
}
${\bf y} \leftarrow \langle y_1,\dots, y_m \rangle$\;
${\bf z} \leftarrow \langle z_1,\dots, z_m \rangle$\;
\If{$P({\bf y}) + P({\bf z}) \ne \delta_1 \cdot 2^n \cdot b_1 + \dots + \delta_m \cdot 2^n \cdot b_m + P({\bf x})$}{
\Return FAIL\;
}
\Return PASS\;
\caption{{\sc RandSplitTwo}}
\label{algo:RandSplitMulti}
\end{algorithm}\DecMargin{1em}

We now use Algorithm~\ref{algo:RandSplitMulti} to create a self-testing algorithm for this linear homomorphism just as we did with the univariate case. Let ${\bf a} = \langle 2^n, \dots, 2^n \rangle$. Then we know that $f(a)=2^n \cdot (b_1 + \dots + b_m)$.

\IncMargin{1em}
\begin{algorithm}[H]
\SetKwFunction{Rand}{Rand}\SetKwFunction{RandSplitTwo}{RandSplitTwo}
\SetKwInOut{Input}{input}\SetKwInOut{Output}{output}
\Input{A program $P$ for computing the function $f$ on a domain $V$ and the values $b_1,\dots,b_m$}
\Output{FAIL or PASS}

\For{$i\leftarrow 1$ \KwTo $k_1$}{
${\bf x} \leftarrow$ \Rand(V)\;
\If{$P({\bf x}) + P({\bf a}-{\bf x}) \ne 2^n \cdot (b_1 + \dots + b_m)$}{
\Return FAIL\;
}
}

\For{$i\leftarrow 1$ \KwTo $k_2$}{
$x \leftarrow$ \Rand(V)\;
\If{\RandSplitTwo$({\bf x},b_1,\dots, b_m,P)$=FAIL}{
\Return FAIL\;
}
}
\Return PASS\;
\caption{{\sc LinearHomomorphismTest}}
\label{algo:LinHomTest}
\end{algorithm}\DecMargin{1em}

\begin{lemma}\label{lemma:linhomalg}
If $P$ is correct on $\le1-\epsilon$ fraction of inputs in $V$, then Algorithm~\ref{algo:LinHomTest} returns FAIL with probability $\ge 3/4$. 
\end{lemma}
\begin{proof}
We note that this is exactly equivalent to the case of Lemma~\ref{lemma:main} and Algorithm~\ref{algo:Test} only with a larger domain. Since the error bounds did not depend on the size of the domain, we see that by Lemma~\ref{lemma:main} we again get that Algorithm~\ref{algo:LinHomTest} returns FAIL with probability $\ge 3/4$.
\end{proof}

Now that we have the algorithm we can show that it is our desired self-testing algorithm.

\begin{proof}[Proof of Theorem~\ref{thm:linhomtestingalg}]
We see easily that if $P$ is correct on all inputs in $V$ that Algorithm~\ref{algo:LinHomTest} always outputs PASS. Then by Lemma~\ref{lemma:linhomalg}, we see that if $P$ is $\epsilon$-far from $f$ that  Algorithm~\ref{algo:LinHomTest} outputs FAIL with probability $\ge 3/4$.

By the same Chernoff bounds from Section~\ref{sec:puttingittogether} we see that we can set $k_1,k_2$ so that $k_1+k_2=O(1/\epsilon)$. The additional computation time depends on how long it takes to compute $f({\bf a})$ in general. The computation time necessary is $O(m)$ shifts and additions. 

 If we assume these shifts and additions take $O(nm)$ time, then the additional computation time is $O(nm/\epsilon)$ as desired. Therefore $P$ is a valid self-testing algorithm for the function $f$ as desired.
\end{proof}

If we let $k_1=96, k_2\approx 709$ as in Section~\ref{sec:puttingittogether}, then since the error bounds are the same as in that section, the Chernoff bounds from the appendix will again give us that each of the two parts will return error with probability 7/8, so we will output FAIL with probability $3/4$ as desired.

\paragraph{Acknowledgements} The authors would like to greatly thank the referees for their comments. The final publication is available at Springer via \url{http://dx.doi.org/10.1007/s00224-015-9639-z}.

\bibliographystyle{plain}
\bibliography{uropproofs}

\section*{Appendices}

\begin{appendix}

\section{Checking Algorithm}

We now describe an algorithm that checks whether the program is correct
when multiplying two $n$-bit numbers $a,b$, rather than if the program is correct in general, in order to prove Theorem~\ref{thm:checker}. The function RandSplit is defined as before:

\IncMargin{1em}
\begin{algorithm}[H]
\SetKwFunction{Rand}{Rand}\SetKwFunction{RandSplit}{RandSplit}\SetKwFunction{Test}{Test}
\SetKwInOut{Input}{input}\SetKwInOut{Output}{output}
\Input{An $n$-bit number $a$, a program $P_b$ computing multiplication by $b$}
\Output{FAIL or PASS}

\lIf{\Test$(P_b)$ = FAIL} \Return FAIL\;

\Return \RandSplit$(a,P_b)$ \label{secondtest}
\caption{{\sc Checker}}
\label{algo:Checker}
\end{algorithm}\DecMargin{1em}

\begin{theorem}Algorithm~\ref{algo:Checker} correctly checks $a \cdot b$ for $a \in D_n$ and outputs FAIL if $P_b$ is incorrect on input $a$ with probability $\ge 3/4$. \end{theorem}

\begin{proof}
It is clear that if the program always answers correctly on the domain $D_n$ then the
checker will output CORRECT.
There are now two cases where $P_b(a)$ is incorrect: the first is $P_b$ is $1/8$-far from $f_b(x)=bx$ and the second is when it is $1/8$-close. If it is $1/8$-far from $f_b$, then we know from Theorem~\ref{thm:testingalg} that when we call Algorithm~\ref{algo:Test} it will return FAIL with probability $\ge 3/4$. If instead $P_b$ is $1/8$-close to $f_b$, then the second test - RandSplit on $a$ in Line~\ref{secondtest} - will be checking $P_b(x)+P_b(a-x)=P_b(a)$ for some $x$. Since $P_b$ is $1/8$-close to $f_b$, we see that the probability that $P_b(x) \ne f_b(x)$ is $\le 1/8$, and the same for the probability that $P_b(a-x) \ne f_b(a-x)$. Therefore by the union bound the probability that $P_b(x) + P_b(a-x) = f_b(x) + f_b(a-x) = f_b(a) \ne P_b(a)$ is $\ge 3/4$, so the second test will return FAIL with probability $\ge 3/4$, as desired. 

We also note that the running time and queries of this algorithm are on the same order as those of Algorithm~\ref{algo:Test}, or $O(n/\epsilon)$ additional computation time and $O(1/\epsilon)$ queries. 
\end{proof}

\section{Chernoff Bounds}

We use Chernoff bounds often to describe the error probabilities of our algorithms. We use the following bound specifically, where $X_1,\dots,X_n$ are random Bernoulli variables with expectation $p$ and $X=\sum X_i$:

\begin{equation*}
Pr[X < (1-\delta)np] \le e^{-\delta^2np/2}
\end{equation*}

where $0 < \delta < 1$. The theorem we use for our algorithms is the following:

\begin{theorem}\label{thm:chernoff}
If the probability that a test correctly detects an error in $P_b$ is $p$, by running the test $6/p=O(1/p)$ times we will detect the error with probability $7/8$.
\end{theorem}
\begin{proof}
Assume we run the test $n$ times. For the $i$th time we run the test, we let $X_i=1$ if an error is detected and $X_i=0$ otherwise. If we run the test $n$ times without detecting error, this means that $\sum X_i = X = 0 < 1$. By the Chernoff bound above, we see that by letting $\delta=(np-1)/np$ that $Pr[X < 1] \le e^{-(np-1)^2/(2np)}$. The probability that $X < 1$ is the same as the probability that we fail to detect an error, which we wish to be $\le 1/8$. Therefore we want $e^{-(np-1)^2/(2np)} \le 1/8$.

\begin{align*}
e^{-(np-1)^2/(2np)} &\le 1/8\\
-\frac{(np-1)^2}{2np} &\le \log 1/8\\
\frac{(np-1)^2}{np} &\ge \log 64 \\
\end{align*}

If $np=6$, then $\frac{(np-1)^2}{np}=25/6 > \log 64$, as desired. Therefore by letting $n=6/p$, the probability that $X<1$ is $\le 1/8$, so we detect an error with probability $7/8$. 

\end{proof}

\end{appendix}

\end{document}